\documentclass[12pt,reqno]{amsart}
\usepackage{amsmath,amsthm,amstext}
\usepackage{amsfonts,amssymb}
\usepackage[mathscr]{eucal}
\usepackage[]{hyperref}
\usepackage{setspace,url}
\usepackage{tikz}
\usepackage{bm}
\usepackage{graphicx}
\usepackage{enumerate}

\newtheorem{theorem}{Theorem}
\newtheorem{lemma}{Lemma}

\newtheorem{proposition}{Proposition}
\newtheorem{remark}{Remark}
\newtheorem{assumption}{Assumption}
\newtheorem{example}{Example}
\numberwithin{equation}{section}

\newcommand{\E}{\mathcal{E}}

\newcommand{\V}{\mathcal{V}}

\newcommand{\Lcal}{\mathcal{L}}

\newcommand{\calL}{\ensuremath{\mathcal{L}}} 

{\begin{trivlist} \item[]{\textbf{Proof} }}%
{\hspace*{\fill}$\rule{.3\baselineskip}{.35\baselineskip}$\end{trivlist}}

\begin{document}

\title[Multi-pulse edge-localized states on quantum graphs]{\bf Multi-pulse edge-localized states on quantum graphs}

\author{Adilbek Kairzhan}
\address{Department of Mathematics, University of Toronto, Toronto, Ontario, Canada}
\email{kairzhan@math.toronto.edu}

\author{Dmitry E. Pelinovsky}
\address{Department of Mathematics, McMaster University, Hamilton, Ontario  L8S 4K1, Canada}
\email{dmpeli@math.mcmaster.ca}

\maketitle

\begin{abstract}
	Edge-localized stationary states of the focusing nonlinear Schr\"{o}dinger equation on a general quantum graph are considered in the limit of large mass. Compared to the previous works, we include arbitrary multi-pulse 
	positive states which approach asymptotically to a composition of $N$ solitons, each sitting on a bounded (pendant, looping, or internal) edge. 
	Not only we prove that such states exist in the limit of large mass, but also we compute the precise 
	Morse index (the number of negative eigenvalues in the corresponding 
	linearized operator). In the case of the edge-localized $N$-soliton states on the pendant and looping edges, we prove that the Morse index is exactly $N$. The technical novelty of this work is achieved by avoiding
	elliptic functions (and related exponentially small scalings) 
	and closing the existence arguments in terms of the Dirichlet-to-Neumann maps for relevant parts of the given graph. 
\end{abstract}

\section{Introduction}

We address standing waves of the focusing 
NLS (nonlinear Schr\"{o}dinger) equation posed on a quantum graph
$\Gamma = \{ \E, \V \}$, where $\E$ is the set of edges 
and $\V$ is the set of vertices (see \cite{Noja,NPSymmetry} for review). The evolution system can be written in the normalized form:
\begin{equation}
\label{nls}
i \Psi_t + \Delta \Psi + 2 |\Psi|^{2} \Psi = 0,
\end{equation}
where the Laplacian $\Delta$ and the nonlinear term are defined 
componentwise on edges $\E$ 
subject to proper boundary conditions on the vertices $\V$ 
(see \cite{BK, Exner} for introduction to linear differential equations 
on quantum graphs). 

The quantum graph $\Gamma = \{ \E, \V \}$ is assumed to consist of a finite number $|\E|$ of bounded and unbounded edges. Enumerating every edge in $\Gamma$ uniquely gives the set 
$\E = \{e_1, e_2, \dots, e_{|\E|}\}$. The function $\Psi$ on $\Gamma$ can be represented as a vector with $|\E|$ components, 
\begin{equation}
\label{psi-vector}
\Psi = \{\psi_1, \psi_2, \dots, \psi_{|\E|}\},
\end{equation}
where $\psi_j$ is defined on the edge $e_j$ only. The function $\Psi$ can be defined in the Hilbert 
space of square-integrable functions $L^2(\Gamma) = \bigoplus_{e \in \E} L^2(e)$. 

Weak (resp. strong) solutions of the NLS time flow (\ref{nls}) 
are well defined in the $L^2$-based Sobolev spaces $H^1(\Gamma)$ (resp. $H^2(\Gamma)$), where $H^{1,2}(\Gamma) = \bigoplus_{e \in \E} H^{1,2}(e)$,
provided the boundary conditions on $\V$ are symmetric. Since $|\E| < \infty$, the NLS time flow (\ref{nls}) is essentially the evolution problem in one spatial dimension, which is globally well-posed both in $H^1(\Gamma)$ and $H^2(\Gamma)$ due to the cubic ($L^2$-subcritical) nonlinearity. 

We consider {\it the natural Neumann--Kirchhoff (NK) boundary conditions} 
at each vertex $v \in \V$ given by
\begin{equation}
\label{kbc}
\left\{ \begin{array}{l}
\Psi \text{ is continuous on } \Gamma, \\
\sum_{e \sim v} \partial \Psi(v) = 0
\text{ for every vertex } v \in \V,
\end{array} \right.
\end{equation}
where the derivatives $\partial$ are directed away from the vertex $v \in \V$ 
and $e \sim v$ denotes the edges $e \in \E$ adjacent to the vertex $v \in \V$.

Consistent with the boundary conditions (\ref{kbc}), 
weak solutions of the NLS equation (\ref{nls}) are defined 
in the energy space $H^1_C(\Gamma) := H^1(\Gamma) \cap C^0(\Gamma)$,
where $C^0(\Gamma)$ denotes the space of functions continuous on all edges in $\E$ and across all vertex points in $\V$. These weak solutions 
conserve the energy and mass functionals given respectively by 
\begin{equation}
\label{quantities}
E(\Psi) = \| \nabla \Psi \|^2_{L^2(\Gamma)} - \| \Psi \|^{4}_{L^{4}(\Gamma)}, \quad Q(\Psi) = \| \Psi \|^2_{L^2(\Gamma)}.
\end{equation}

Standing waves of the NLS equation (\ref{nls})
are given by the solutions of the form 
$\Psi(t,x) = \Phi(x) e^{-i \omega t}$,
where $\Phi \in H^1_C(\Gamma)$ is a weak solution of 
the stationary NLS equation
\begin{equation}
\label{nls-stat}
\omega \Phi =  - \Delta \Phi - 2 |\Phi|^{2} \Phi,
\end{equation}
for a given $\omega \in \mathbb{R}$. By bootstrapping arguments, 
every weak solution of the stationary NLS equation (\ref{nls-stat}) 
is also a strong solution satisfying the natural NK conditions (\ref{kbc}). Hence, we can use the vector representation $\Phi = (\phi_1, \phi_2, \dots, \phi_{|\E|})$ and rewrite the stationary NLS equation (\ref{nls-stat}) on every edge $e_j$ in $\E$ 
as a collection of differential equations: 
\begin{equation}
\label{nls-stat-j}
\omega \phi_j(x) =  - \phi_j''(x) - 2 |\phi_j(x)|^{2} \phi_j(x), \quad 
x \in e_j,
\end{equation}
which satisfy the boundary conditions (\ref{kbc}) on every vertex $v \in \V$.
We write $\Phi \in H^2_{\rm NK}(\Gamma)$ if $\Phi \in H^2(\Gamma)$
satisfies the NK conditions (\ref{kbc}). Only 
{\em real-valued} solutions of the stationary NLS equation are considered 
but we write the modulus sign for easy generalizations.

Among all possible real-valued solutions of the stationary NLS equation (\ref{nls-stat}), we are particularly interested
in the {\em positive edge-localized states} which satisfy the following conditions:
\begin{itemize}
\item $\Phi(x) > 0$ for every $x \in \Gamma$;
\item on each bounded edge $e_j \in \E$, there is at most one local 
critical point of $\phi_j$ (either maximum or minimum) inside the edge;
\item on each unbounded edge $e_j \in \E$, the function 
$\phi_j$ is monotonically decreasing and has exponential decay to $0$ at infinity.
\end{itemize}
Depending on the topological properties of the quantum graph $\Gamma$, 
the positive edge-localized states could become {\em the ground state}, 
the state of the least energy $E(\Psi)$ at fixed mass $Q(\Psi)$ \cite{AdamiCV,AdamiJFA,AST17}.

Although the set of conditions on $\Phi$ seems to be restrictive, the positive edge-localized states exist {\it in the limit of large mass} \cite{AST,BMP,D20,KS20}.
This limit for large $\mu = Q(\Phi)$ can be recast as the limit of large negative $\omega$ in the stationary NLS equation (\ref{nls-stat}) for the   cubic ($L^2$-subcritical) nonlinearity. Existence of such edge-localized 
states was confirmed analytically and/or numerically for the tadpole graph \cite{AdamiJFA,CFN,Dovetta,NP,NPS}, the dumbbell graph \cite{G,MP}, and the flower graph \cite{KMPX}. The importance of the positive edge-localized states is motivated by their (possible) orbital stability in the NLS time flow.

\begin{figure}[htbp] 
	\centering
	\includegraphics[width=1.75in, height = 1in]{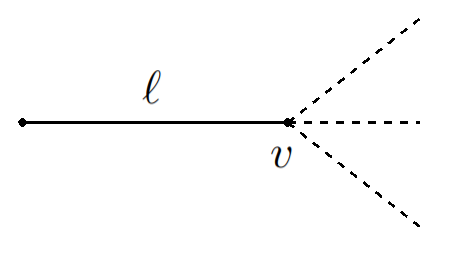} \hspace{0.5cm}
	\includegraphics[width=1.75in, height = 1in]{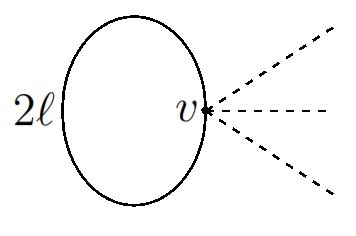}
	\hspace{0.5cm}
	\includegraphics[width=1.75in, height = 1in]{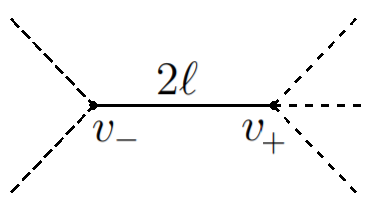}
	\caption{Schematic illustration of a pendant edge (left), a looping edge (middle), and an internal edge (right). }
	\label{fig-edges}
\end{figure}

Bounded edges of $\E$ are divided into pendant, looping, and internal edges 
according to the following classification (see illustrations on Fig. \ref{fig-edges}): 
\begin{itemize}
	\item {\em A pendant edge} of length $\ell$ is associated 
	with the segment $[0,\ell]$, where the left end is isolated from 
	the rest of $\Gamma$ subject to the Neumann boundary condition 
	and the right end is connected with the rest of $\Gamma$ at a vertex $v$.
	
	\item {\em A looping edge} of length $2 \ell$ is associated with the segment $[-\ell,\ell]$, 
	where both ends at connected to the rest of $\Gamma$ at a single vertex $v$, 
	hence contributing twice to the derivative condition in (\ref{kbc}).
	
	\item {\em An internal edge} of length $2 \ell$ is associated with the segment $[-\ell,\ell]$, 
	where different ends at connected to the rest of $\Gamma$ at two different vertices $v_-$ and $v_+$. 
\end{itemize}
Each unbounded edge is associated with the half-line $[0, \infty)$. 
The edge-localized states considered in \cite{BMP} consist of a single 
large-amplitude component $\phi_e$ on a fixed looping or internal 
edge $e$ of length $2 \ell$ such that $\phi_e$ has a single local maximum inside $e$, monotone from its maximum to the vertices of $e$, and 
concentrated on $e$ in the following sense
\begin{equation}
\label{eq:L2_proportion}
\frac{\left\|\Phi \right\|_{L^2(e)}}
{\left\|\Phi\right\|_{L^2(\Gamma)}}
\geq 1-Ce^{-2\epsilon\ell},
\end{equation}
where $\epsilon := \sqrt{|\omega|}$ is a large parameter 
and the constant $C$ is independent of $\epsilon$. Moreover, $\Phi$ has no internal maxima on the remainder of graph $\Gamma \backslash \{e\}$. 
The same result also holds for the pendant edge $e$ except that $\phi_e$ has a single local maximum at the terminal vertex. The edge-localized states considered in \cite{BMP} are illustrated on Fig. \ref{fig-states}.

\begin{figure}[htbp] 
	\centering
	\includegraphics[width=2in, height = 1.5in]{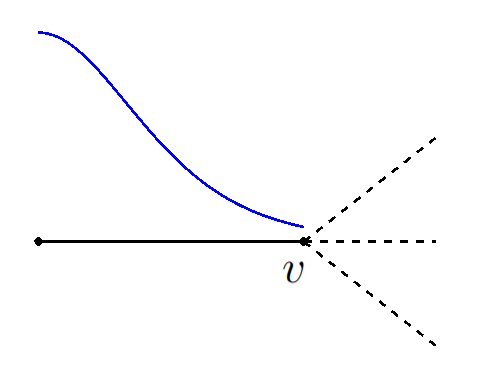} 
	\hspace{0.5cm}
	\includegraphics[width=2in, height = 1.5in]{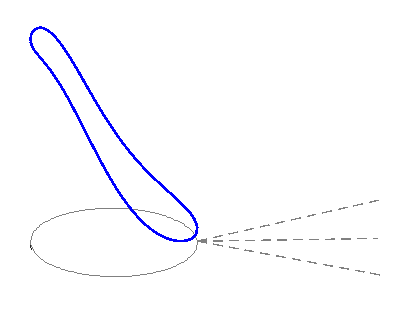}
	\hspace{0.5cm}
	\includegraphics[width=2in, height = 1.5in]{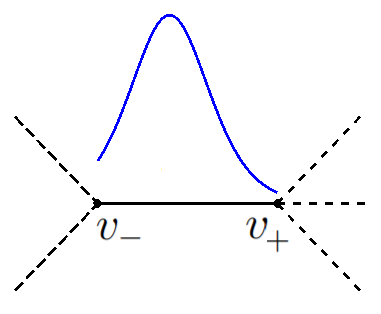}
	\caption{Schematic illustration of edge-localized states on a pendant (left), a looping edge (middle), and an internal edge (right). }
	\label{fig-states}
\end{figure}

The construction of the edge-localized states in \cite{BMP} relies 
on the properties of elliptic functions and on careful rescaling 
of exponentially small terms for the elliptic modulus. In addition, 
the states of lowest energy at fixed mass were analyzed 
in \cite{BMP} by comparing the exponentially small terms in the 
expansion of $\mu := Q(\Phi)$ in $\epsilon := \sqrt{|\omega|}$.

A similar result on the edge-localized states in the large-mass limit was obtained  independently in \cite{AST} using variational methods (for unbounded graphs only). Since the edge-localized states 
were identified in \cite{AST} as local energy minimizers in the restricted space of functions in $H^1_C(\Gamma)$ that attain their maximum on a given edge $e$, the Morse index for such states is exactly one, where {\em the Morse index} is the number of negative eigenvalues of the linearized operator $\mathcal{L} :  H^2_{\rm NK}(\Gamma) \subset L^2(\Gamma) \to L^2(\Gamma)$ given by 
\begin{equation}
\label{Lplus}
\Lcal = -\Delta - \omega - 6 |\Phi|^{2}.
\end{equation}

If the graph $\Gamma$ is unbounded, $\Phi$ is bounded and decays to $0$ at infinity exponentially fast, and $\omega < 0$, then the linearized operator  
$\mathcal{L}$ is self-adjoint and the absolutely continuous part of the spectrum is strictly 
positive and bounded away from zero by $|\omega|$ by Weyl's theorem on 
essential spectrum. Therefore, the Morse index is well-defined as the finite 
number of negative eigenvalues of $\mathcal{L}$ accounting for 
their finite multiplicity. We denote this number by $n(\mathcal{L})$. 
We also denote the multiplicity of the zero eigenvalue of $\mathcal{L}$ 
by $z(\mathcal{L})$. 

If the graph $\Gamma$ is bounded, then the spectrum of $\mathcal{L}$ 
is purely discrete and the numbers $n(\mathcal{L})$ and $z(\mathcal{L})$ 
are again well-defined since $\mathcal{L}$ is bounded from below.

Additional results related to the edge-localized states in the limit of large energy were obtained in \cite{D20}, where bounded graphs with the pendant edges were considered and convergence of the edge-localized states to the half-solitons was proven. Multi-pulse states were also studied in \cite{D20}, all pulses localize at the terminal vertices. Edge-localized states were considered in \cite{KS20} by recasting the existence problem to the semi-classical limit of an elliptic problem. It was proven in \cite{KS20} that the location of the edge-localized state with a single maximum as the ground state of energy at fixed mass is determined by the longest pendant edge of a bounded graph (or the longest internal edge if no pendant and looping edges are present). These results are included in those in \cite{BMP} but they are obtained in a more general setting in \cite{KS20}. 

{\em The main result of this work concerns the construction of the 
multi-pulse edge-localized states in the limit of 
large mass (large negative $\omega$).} As $\omega \to -\infty$, the multi-pulse states approach asymptotically a composition of $N$ solitons, 
each sitting on a bounded (pendant, looping, or internal) edge. 
The edges of the graph, where the multi-pulse state is localized, 
must satisfy the following non-degeneracy assumptions.

\begin{assumption}
	\label{assumption-non-degeneracy}
Let $\E_N := \{ e_1,e_2,\dots,e_N\}$ be the set of $N$ edges 
connected to $\Gamma \backslash \E_N$ at $|B|$ boundary vertices $\V_B = \{ v_1, v_2, \dots,v_{|B|}\}$. Let $\ell_{\rm min}$ be the length of the shortest edge in $\Gamma \backslash \E_N$ and $\ell_{j,{\rm min}}$ be the minimal half-length of the looping and internal edges or the minimal length of the pendant edges adjacent to the boundary vertex $v_j \in \V_B$ from $\E_N$, $1 \leq j \leq |B|$. 
The edge lengths must satisfy the constraints:	
	\begin{equation}
	\label{constraint-on-length-1}
	\ell_{\rm min} + \min_{1 \leq i \leq |B|} \ell_{i,{\rm min}} > \ell_{j,{\rm min}} 
	\quad 1 \leq j \leq |B|
	\end{equation}
	and
	\begin{equation}
	\label{constraint-on-length-2}
	3 \min_{1 \leq i \leq |B|} \ell_{i,{\rm min}} > \ell_{j,{\rm min}},
	\quad 1 \leq j \leq |B|.
	\end{equation}
\end{assumption}

\begin{assumption}
	\label{assumption-internal-edge}
Each internal edge has no common boundary vertices with other internal edges in $\E_N$ and its half-length  is strictly minimal 
among the half-lengths of the looping edges and the lengths of the pendant edges 
in $\E_N$ adjacent to its two boundary vertices.
\end{assumption}

\begin{remark}
	\label{remark-degree}
	In the definition of the edge lengths, we assume that each vertex has degree $3$ or higher, except for the terminal vertices of the pendant edges which have degree $1$. In other words, we do not allow fake vertices of degree $2$ in the graph $\Gamma$.
\end{remark}

The following theorem gives the existence result.

\begin{theorem}
	\label{theorem-main1}
	Let $\Gamma = \{ \E, \V \}$ be a graph with finitely many edges and
satisfying NK conditions at its vertices.  For any $N$ edges $\E_N := \{ e_1,e_2,\dots,e_N\}$ 
of finite lengths satisfying Assumptions \ref{assumption-non-degeneracy} and \ref{assumption-internal-edge} and for large enough $\epsilon := \sqrt{-\omega}$
there exists a positive edge-localized state $\Phi$ with the following properties:
\begin{enumerate}
	\item $\Phi |_{e_j}$ has a single local maximum and is monotone 
	from its maximum to the end vertices of $e_j \in \E_N$, for $1 \leq j \leq N$;
	\item $\Phi$ has no internal maxima on the remainder 
	of the graph $\Gamma \backslash \E_N$; 
	\item $\Phi$ is concentrated on $\E_N$ in the following sense 
	\begin{equation}
	\label{eq:L2}
	\frac{\left\|\Phi \right\|_{L^2(\Gamma \backslash \E_N)}}
	{\left\|\Phi\right\|_{L^2(\E_N)}}
	\leq  C e^{-\epsilon \ell_N},
	\end{equation}		
\end{enumerate}
where the constant $C$ is independent of $\epsilon$ 
and $\ell_N = \min\{\ell_{1,{\rm min}},\ell_{2,{\rm min}},\dots,\ell_{|B|,{\rm min}}\}$.
\end{theorem}

\begin{remark}
	\label{remark-th-1}
	For the pendant edge $e_j \in \E_N$, the maximum of $\Phi |_{e_j}$ occurs 
	at the left edge with the Neumann boundary conditions. 
	For the looping edge $e_j \in \E_N$, the maximum of $\Phi |_{e_j}$ occurs exactly at the middle point. For the internal edge $e_j \in \E_N$, the maximum of $\Phi |_{e_j}$ 
	is generally shifted from the middle point. The non-degeneracy condition 
	in Assumption \ref{assumption-internal-edge} allows us to control 
 	 the shift of the maximum of $\Phi |_{e_j}$ at the internal edge. 
 	 See Fig. \ref{fig-states} for illustration of edge-localized states on each bounded edge in $\E_N$.
\end{remark}

\begin{remark}
		\label{remark-th-2}
	Condition (\ref{eq:L2}) is equivalent to (\ref{eq:L2_proportion}) in the case $N = 1$ and $\E_N= \{ e \}$. Indeed, it follows from (\ref{eq:L2}) with $\ell_{N=1} = \ell$ that 
\begin{equation}
\label{L2norm-1}
	\| \Phi \|^2_{L^2(\Gamma)} = \| \Phi \|^2_{L^2(e)} + \| \Phi \|^2_{L^2(\Gamma \backslash \{e\})} \leq (1 + C^2 e^{-2\epsilon \ell}) \| \Phi \|^2_{L^2(e)},
\end{equation}
	which yields (\ref{eq:L2_proportion}). It follows from Theorems 3.1, 3.3, and 3.5 in \cite{BMP} that 
\begin{equation}
\label{L2norm-2}
	| \| \Phi \|^2_{L^2(e)} - \epsilon | \leq C \epsilon^2 e^{-2\epsilon \ell}
\end{equation}
for the pendant edge and 
\begin{equation}
\label{L2norm-3}
	| \| \Phi \|^2_{L^2(e)} - 2 \epsilon | \leq C \epsilon^2 e^{-2\epsilon \ell}
\end{equation}
for the looping and internal edges, where the constant $C$ is independent of $\epsilon$. In the case of $N$ edges,  we have	
\begin{equation}
\label{L2norm-4}
	\| \Phi \|^2_{L^2(\E_N)} = \sum_{j=1}^N \| \Phi \|^2_{L^2(e_j)}
\end{equation}
	with similar estimates for $\| \Phi \|^2_{L^2(e_j)}$, where $e_j \in \E_N$.
\end{remark}

\begin{remark}
		\label{remark-th-3}
Compared to the work in \cite{BMP}, we do not use 
elliptic functions and exponentially small scalings, 
which makes our results more general and the proofs simpler. 
We partition the graph $\Gamma$ into $\E_N$ and $\Gamma \backslash \E_N$ 
and reduce the existence problem to a system of equations for 
Dirichlet data on the boundary vertices in $\V_B$. 
Then, we show that all these 
equations can be solved independently of each other under 
the constraints (\ref{constraint-on-length-1}) and (\ref{constraint-on-length-2}) on the lengths of edges 
in Assumption \ref{assumption-non-degeneracy}. 
These constraints
provide compatability of asymptotic solutions for large $\epsilon$ 
and allow us to ignore the Dirichlet data at other boundary vertices.
\end{remark}

Our second main result is the precise characterization of 
the Morse index $n(\mathcal{L})$ and degeneracy index $z(\mathcal{L})$ 
of the multi-pulse edge-localized states in the limit of 
large mass (large negative $\omega$). 
This characterization was not provided in \cite{BMP} even in the case 
of $N = 1$. The following theorem gives the result 
when the set $\E_N$ does not include internal edges.

\begin{theorem}
	\label{theorem-main2}
	Let $\Phi$ be the positive $N$-pulse edge-localized state of Theorem \ref{theorem-main1} for large enough $\epsilon$ 
	and assume that the set $\E_N$ contains only pendant and looping edges. 
	Then, $n(\mathcal{L}) = N$ and $z(\mathcal{L}) = 0$, 
	where $\mathcal{L}$ is the linearized operator in (\ref{Lplus}).
\end{theorem}

\begin{remark}
To obtain the exact count of the Morse index and the multiplicity of the zero eigenvalue of $\calL$, we use a homotopy argument relating the non-positive spectra of the operators $\calL: H_{\rm NK}^2(\Gamma) \subset L^2(\Gamma) \to L^2(\Gamma)$ and $\calL_D: H_{\rm D}^2(\Gamma) \subset L^2(\Gamma) \to L^2(\Gamma)$ where $\calL_D$ has the same differential representation as $\calL$ and $H_{\rm D}^2(\Gamma)$ differs from $H_{\rm NK}^2(\Gamma)$ by the Dirichlet conditions at the boundary vertices in $\V_B$ instead of the NK conditions. The technique resembles the surgery principle widely used in the spectral analysis of differential operators on graphs \cite{BKKM,LS}.
\end{remark}

\begin{remark}
		\label{remark-th-4}
The implication of Theorem \ref{theorem-main2} to the time evolution 
of perturbations to the positive $N$-pulse edge-localized states 
of Theorem \ref{theorem-main1} is that these states are orbitally 
unstable under the NLS time flow (\ref{nls}) if $N \geq 2$. 
This follows from an easy application of the main results of \cite{Gr} and \cite{SS}. In agreement with the variational 
characterization of the single-pulse states on unbounded graphs $\Gamma$ in \cite{AST},  the positive edge-localized states with $N = 1$ 
are orbitally stable. This follows via the standard 
orbital stability theory \cite{GSS} from the monotonicity of the mapping 
\begin{equation}
\label{L2-mapping}
\epsilon \mapsto \left\|\Phi\right\|^2_{L^2(\E_N)}
\end{equation}
for large $\epsilon$ and the exponential smallness of $\left\|\Phi\right\|_{L^2(\E_N)}$ due to the decomposition (\ref{L2norm-4}) with bounds  (\ref{eq:L2}), (\ref{L2norm-2}) and (\ref{L2norm-3}). Note that it was proven in \cite{BMP} that the mapping (\ref{L2-mapping}) is $C^1$.
\end{remark}

\begin{remark}
If an internal edge is present in the set $\E_N$, then the explicit count of the Morse index in Theorem \ref{theorem-main2} is no longer applicable, as shown
in Example \ref{example-counter}. Therefore, characterization of the Morse index 
and the multiplicity of the zero eigenvalue of $\mathcal{L}$ is still an open problem in the case of multi-pulse states localized on internal edges.
\end{remark}

The paper is organized as follows. Section \ref{sec-2} reviews 
preliminary results from \cite{BMP} and \cite{KMPX}. 
Section \ref{sec-proof} and \ref{sec-proof-2} give the proof 
of Theorems \ref{theorem-main1} and \ref{theorem-main2} respectively. 
Section \ref{sec-4} contains examples of flower, dumbbell, and single-interval graphs where Theorems \ref{theorem-main1} and \ref{theorem-main2} are applicable and discuss counter-examples when assumptions of Theorem \ref{theorem-main2} are violated.

\section{Preliminary results}
\label{sec-2}

Since $\epsilon := \sqrt{-\omega}$ is considered to be large, 
we rescale the stationary NLS equation (\ref{nls-stat})
with the transformation $\Phi(x) = \epsilon U(\epsilon x)$, $x \in \Gamma$. 
As a result, the graph $\Gamma = \{ \E, \V \}$ is transformed to the 
$\epsilon$-scaled graph $\Gamma_{\epsilon} = \{ \E_{\epsilon}, \V \}$ 
for which every bounded edge $e \in \E$ of length $\ell_e$ 
transforms to the edge $e_{\epsilon} \in \E_{\epsilon}$ of length $\epsilon \ell_e$ but the unbounded edge $e \in \E$ remains the same as $e \in \E_{\epsilon}$. The stationary NLS equation (\ref{nls-stat}) 
is rewritten in the parameter-free form 
\begin{equation}
\label{statNLS-limit-graph}
(1 -\Delta) U - 2 |U|^2 U = 0.
\end{equation}
Writing $U = (u_1,u_2,\ldots,u_{|\E|})$ 
gives a collection of differential equations 
\begin{equation}
\label{nls-scaled}
-u_j''(z) + u_j(z) - 2 |u_j(z)|^{2} u_j(z) = 0, \quad 
z \in e_{j,\epsilon}, 
\end{equation}
subject to the boundary conditions (\ref{kbc}) at vertices $v \in \V$.

We select $N$ edges $\E_N := \{ e_1,e_2,\ldots, e_N\}$ with $N \leq |\E|$ in the original graph $\Gamma$ 
and partition the graph into two parts $\E_N$ and $\Gamma \backslash \E_N$. 
The two graphs intersect at the vertices in the set $\V_B := \{ v_1,v_2,\ldots, v_{|B|} \}$, which are referred to as {\em boundary vertices}. 
After rescaling of $\Gamma$ to $\Gamma_{\epsilon}$, 
we obtain $\E_{N,\epsilon}$ and $\Gamma_{\epsilon} \backslash \E_{N,\epsilon}$.

We introduce the Dirichlet data on the boundary vertices 
$\vec{p} = (p_1,p_2, \ldots, p_{|B|})$ and the Neumann data 
$\vec{q} = (q_1,q_2, \ldots, q_{|B|})$ for the graph $\Gamma_{\epsilon} \backslash \E_{N,\epsilon}$ with 
\begin{equation}
\label{Neumann_data}
p_j := u_{e \sim v_j}(v_j), \quad 
q_j := \sum_{e \sim v_j} \partial u_e(v_j), \quad v_j \in \V_B,
\end{equation}
where the derivatives $\partial$ are directed away from $\Gamma_{\epsilon} \backslash \E_{N,\epsilon}$ and $e \sim v_j$ lists all edges $e \in \Gamma_{\epsilon} \backslash \E_{N,\epsilon}$ incident to the vertex $v_j$. 
We are looking for solutions of the differential equations 
(\ref{nls-scaled}) such that 
\begin{equation}
\label{bound-down}
\sup_{z \in e_{j,\epsilon}} |u_j(z)| < \frac{1}{\sqrt{2}}, \quad 
e_{j,\epsilon} \in \Gamma_{\epsilon} \backslash \E_{N,\epsilon},
\end{equation}
and 
\begin{equation}
\label{bound-up}
\sup_{z \in e_{j,\epsilon}} |u_j(z)| > \frac{1}{\sqrt{2}}, \quad 
e_{j,\epsilon} \in \E_{N,\epsilon}
\end{equation}
where $\frac{1}{\sqrt{2}}$ is the 
constant solution of the differential equations in (\ref{nls-scaled}).

The following two lemmas were proven in \cite{BMP} (Theorem 2.9 and Lemma 2.12). 

\begin{lemma}
	\label{lemma-1}
There exist $C_0 > 0$, $p_0 > 0$, and $\epsilon_0>0$ such that for every
$\vec{p}$ with $\| \vec{p}\| < p_0$ and every
	$\epsilon > \epsilon_0$, there exists a solution $U \in H^2_{\rm NK}(\Gamma_{\epsilon} \backslash \E_{N,\epsilon})$
	to the stationary NLS equation (\ref{statNLS-limit-graph}) 
	on $\Gamma_{\epsilon} \backslash \E_{N,\epsilon}$ subject 
	to the Dirichlet data on $\V_B$ which
	is unique among functions satisfying \eqref{bound-down}.
	The solution satisfies the estimate
	\begin{equation}
	\label{eq:nlin_DTN_solution}
	\| U \|_{H^2(\Gamma_{\epsilon} \backslash \E_{N,\epsilon})} \leq C_0 \|\vec{p}\|,
	\end{equation}
	while its Neumann data satisfies
	\begin{equation}
	\label{eq:nlin_DTN_value}
	|q_j - D_j p_j| \leq C_0 \left( \|\vec{p} \| e^{-\epsilon \ell_{\rm min}} +
	\|\vec{p}\|^3 \right),	\qquad 1 \leq j \leq |B|,
	\end{equation}
	where $D_j$ is the degree of the $j$-th boundary vertex 
	in $\Gamma_{\epsilon} \backslash \E_{N,\epsilon}$ and $\ell_{\rm min}$
	is the length of the shortest edge in $\Gamma \backslash \E_N$.
	Furthermore, if $p_j \geq 0$ for every $j$, then $U(z) \geq 0$ for
	all $z \in \Gamma_{\epsilon} \backslash \E_{N,\epsilon}$ and 
	$U$ has no internal local maxima in $\Gamma_{\epsilon} \backslash \E_{N,\epsilon}$.
\end{lemma}

\begin{lemma}
	\label{lemma-2}
	There exist $C_0 > 0$, $p_0 > 0$, and $\epsilon_0>0$ such that for every
	$p \in (0,p_0)$ and every $\epsilon > \epsilon_0$, there exists a real solution $u \in H^2(0,\epsilon \ell)$ to the differential equation $-u'' + u - 2 u^3 = 0$ satisfying $u'(0) = 0$ and $u(\epsilon \ell) = p$, which is unique among positive and decreasing functions satisfying \eqref{bound-up}.	The solution satisfies $u'(\epsilon \ell) < 0$ and
	\begin{equation}
\left| u'(\epsilon \ell) - u(\epsilon \ell) + 4 e^{-\epsilon \ell} \right|
\leq C_0 \epsilon e^{-3 \epsilon \ell}.
\label{eq:single-bump-DtN}
	\end{equation}
\end{lemma}

\begin{remark}
	\label{remark-1}
	The	$C^1$ property of Neumann data $\vec{q}$ and $u'(\epsilon \ell)$ 
	with respect to parameter $\epsilon$ and Dirichlet data $\vec{p}$ and $u(\epsilon \ell) = p$ in Lemmas \ref{lemma-1} and \ref{lemma-2} respectively was established in \cite{BMP} with similar exponentially small estimates. We will not write this property expalicitly but will 
	use the $C^1$ property in the application of the implicit function theorem. 
\end{remark}
	
\begin{remark}
	The solution $u(z)$ in Lemma \ref{lemma-2} was represented by elliptic functions and its dependence on the elliptic modulus was also 
	studied in Lemma 2.12 of \cite{BMP}.
	This information is not used in the present work.
\end{remark}

For the solution $u \in H^2(0,\epsilon \ell)$ to the differential equation $-u'' + u - 2 u^3 = 0$ in Lemma \ref{lemma-2}, we write the boundary conditions as
\begin{equation}
\label{u-sol-bound-values}
p_+ := u(0), \quad 0 = u'(0), \quad 
p := u(\epsilon \ell), \quad q := -u'(\epsilon \ell),
\end{equation}
where $p_+ \in (p,1)$, $p > 0$, and $q > 0$.
Denoting $v(z) := u'(z)$, the pair $(u(z), v(z))$ for all $z\in (0, \epsilon \ell)$ stays on the invariant curve 
\begin{equation}
\label{energy-level}
E_{\beta} := \{(u,v): \;\; v^2 - u^2 + u^4 = \beta\}
\end{equation}
with some constant $\beta > -\frac{1}{4}$. In particular, in the view of (\ref{u-sol-bound-values}), the constant $\beta$ satisfies $\beta = q^2 - p^2 + p^4 = -p_+^2 + p_+^4$. Fig. \ref{fig-period} shows the $(u,v)$ phase plane 
and the invariant curve between the points $(p_+,0)$ and $(p,-q)$. The period function $T_+(p,q)$ defined by 
\begin{equation}
\label{period}
T_+(p,q) := \int_{p}^{p_+} \frac{du}{\sqrt{\beta + u^2 - u^4}}
\end{equation}
gives the $z$-length of the solution $u(z)$ obtained along the invariant curve $E_\beta$ between points $(p_+,0)$ and $(p,-q)$. For the solution $u \in H^2(0,\epsilon \ell)$ in Lemma \ref{lemma-2}, the relations (\ref{u-sol-bound-values}) can be written as 
\begin{equation*}
T_+(p, q) = \epsilon \ell, \quad
u(T_+(p, q)) = p, \quad u'(T_+(p, q)) = -q.
\end{equation*}
Compared to the standard terminology, see, e.g., \cite{Vill1}, where the period function is introduced for the fundamental period of the periodic function $u(z)$ along the integral curve $E_{\beta}$, the period function $T_+(p,q)$ given by (\ref{period}) corresponds to a part of the integral curve $E_{\beta}$.

\begin{figure}[htbp] 
	\centering
	\includegraphics[width=6in, height = 4in]{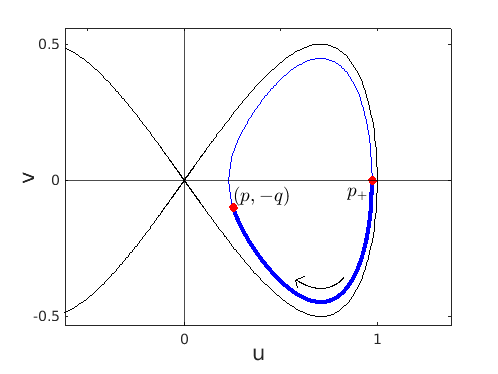}
	\caption{Phase plane $(u,v)$ for the differential equation $-u'' + u - 2 u^3 = 0$ showing the homoclinic loop for $\beta = 0$, the integral curve $E_{\beta}$ for $\beta \in (-\frac{1}{4},0)$, and the part of the integral curve between points $(p_+,0)$ and $(p,-q)$.}
	\label{fig-period}
\end{figure}

The following two lemmas partially proven in \cite{KMPX} 
describe properties of the period function $T_+(p,q)$ and the solution $u(z)$ in Lemma \ref{lemma-2}.

\begin{lemma}
\label{lemma-3}
There is small $\delta > 0$ such that 
\begin{equation}
\label{period-property}
\partial_p T_+(p,q) < 0\quad \mbox{\rm and} \quad \partial_q T_+(p,q) < 0
\end{equation}
for every $p \in (0,\delta)$ and $q \in (0,\delta)$.
\end{lemma}

\begin{proof}
It follows from the bound (\ref{eq:single-bump-DtN}) that 
\begin{equation}
\label{period-asymptotics}
T_+(p,q) = -\ln\left(\frac{p+q}{4}\right) + \mathcal{O}\left( 
(p^2+q^2)(|\ln p| + |\ln q|)\right) \quad \mbox{\rm as} \quad p,q \to 0.
\end{equation}
Taking derivatives of (\ref{period-asymptotics}) in $p$ and $q$ 
justifies the inequalities in (\ref{period-property}). 
\end{proof}

\begin{remark}
	The second property in (\ref{period-property}) is also proven 
	in Lemma 3.6 of \cite{KMPX} by using more complicated analysis 
	of the integrals in (\ref{period}).
\end{remark}

\begin{lemma}
	\label{bound-up-s-vanishing-points}
	Let $u$ be a real, positive, and decreasing solution to the nonlinear equation $-u'' + u - 2 u^3 = 0$ on $(0,T_+(p,q))$ satisfying 	
\begin{equation}
\label{bound-up-u-period-relation}
	u(0) = p_+, \quad u'(0) = 0, \quad u(T_+(p,q)) = p, \quad u'(T_+(p,q)) = -q
\end{equation} 
with some $p > 0$ and $q > 0$, where $T_+(p,q)$ and $p_+$ are defined in (\ref{period}). Then, 
	the general solution of the linearized equation $- w''(z) + w(z) - 6 u^2(z) w(z) = 0$ is given by 
\begin{equation}
\label{superposition}
	w(z) = A u'(z) + B s(z),
\end{equation}
	where $A$ and $B$ are arbitrary parameters and $s(z)$ satisfies 
	$s(0) \neq 0$, $s'(0) = 0$. Moreover, for sufficiently small $p$ and $q$, $u'$ is positive on $(0,T_+(p,q))$, $s$ vanishes at exactly one point on the interval $(0,T_+(p,q))$ and both functions satisfy
	\begin{equation}
	\label{bound-up-weyl-tit-s}
	\frac{u''(T_+(p, q))}{u'(T_+(p, q))} > 0, \quad 
	\frac{s'(T_+(p,q))}{s(T_+(p,q))} > 0.
	\end{equation}
\end{lemma}

\begin{proof}
	It follows from properties of $u(z)$ in Lemma \ref{lemma-2} that $u'(z) > 0$ for $z \in (0,T_+(p,q))$. The first inequality in (\ref{bound-up-weyl-tit-s}) follows from 
\begin{equation}
\label{bound-down-odd-weyl-tit}
\frac{u''(T_+(p, q))}{u'(T_+(p, q))} = \frac{p(1-2p^2)}{q} > 0,
\end{equation} 
for $p \in (0,\frac{1}{\sqrt{2}})$ and $q > 0$. 

In order to study the decomposition (\ref{superposition}), we introduce $s(z) := \partial_q u(z)$ and $r(z) := \partial_p u(z)$. 
The two functions satisfy the following boundary conditions obtained after differentiating (\ref{bound-up-u-period-relation}) with respect to $p$ and $q$:
\begin{equation*}
\left\{ \begin{array}{l} 
s(T_+(p,q)) = q \partial_q T_+,  \\
s'(T_+(p,q)) = -1 - p (1-2p^2) \partial_q T_+
\end{array} \right.
\end{equation*}
and
\begin{equation*}
\left\{ \begin{array}{l}
r(T_+(p,q)) = 1 + q \partial_p T_+,  \\
r'(T_+(p,q)) = -p (1-2p^2) \partial_p T_+.\end{array} \right.
\end{equation*}
Since $s(0) = \frac{-q}{p_+(1-2p_+^2)}$, $r(0) = \frac{p(1-2p^2)}{p_+(1-2p_+^2)}$, and $s'(0) = r'(0) = 0$, 
uniqueness of solutions of the differential equations implies that 
the two functions are related by 
\begin{equation}
\label{relation-s-r}
s(z) =  - \frac{q}{p(1-2p^2)} r(z).
\end{equation}
It follows from the boundary conditions at $z = T_+(p,q)$ and relation (\ref{relation-s-r}) that 
\begin{equation}
\label{bound-up-s-left-boundary-values}
\left\{ \begin{array}{l} s(T_+(p, q)) = q \partial_q T_+, \\
s'(T_+(p, q)) = q \partial_p T_+, 
\end{array} \right.
\end{equation}
so that 
$$
\frac{s'(T_+(p, q))}{s(T_+(p,q))} = \frac{\partial_p T_+(p,q)}{\partial_q T_+(p,q)} > 0,
$$
where positivity follows from (\ref{period-property}) for sufficiently small $p$ and $q$. 

In order to prove that $s$ vanishes at exactly one point on $(0,T_+(p,q))$, we consider the invariant curve $E_\beta$ given by (\ref{energy-level}) and the period function $T_+(p,q)$ given by (\ref{period}). Every $z \in (0,T_+(p,q))$ can be represented by $z = T_+(P, Q)$ for some point $(P, Q) \in E_\beta$ with $P \in (p, p_+)$. By Lemma 3.8 in \cite{KMPX}, we have $s(z) = 0$ if and only if $\partial_Q T_+(P, Q) = 0$, which also follows from the first equation 
in system (\ref{bound-up-s-left-boundary-values}). Then, by Lemmas 3.6, 3.7, 3.9, 3.10 in \cite{KMPX}, for sufficiently small $p$ and $q$ there is exactly one point $(P, Q) \in E_\beta$  with $P\in (p, p_+)$, where $s(z) = s(T_+(P,Q)) = 0$. 
\end{proof}

\section{Existence of multi-pulse loop-localized states} 
\label{sec-proof}

Here we explain the asymptotic construction of the multi-pulse loop-localized 
states and give the proof of Theorem \ref{theorem-main1}. 

For each boundary vertex $v_j \in \V_B$ with $1 \leq j \leq |B|$, 
we use the Dirichlet data $p_j$ as the unknown variable
and write the flux boundary condition to determine $p_j$. 
The main advantage of this method is that the value of $p_j$ 
can be found independently from the conditions at the other boundary vertices. 

Fix $j$ in $1 \leq j \leq |B|$.
Assume that the boundary vertex $v_j \in \V_B$ is in contact with 
$K_j$ pendants, $L_j$ loops, and $M_j$ internal edges 
from the set $\E_{N,\epsilon}$. We denote the 
corresponding sets by $\E_{j,{\rm pend}}$, $\E_{j,{\rm loop}}$, 
and $\E_{j,{\rm int}}$ respectively. Furthermore, we divide $\E_{j,{\rm int}}$ 
into $\E^-_{j,{\rm int}}$ and $\E^+_{j,{\rm int}}$ depending 
on whether $v_j$ is the left or right vertex of the internal edge $e \in \E_{j,{\rm int}}$, respectively. The large solutions of Lemma \ref{lemma-2} 
are centered at $0$ on the pendants and loops 
but centered at an unknown point $a_e \in (-\epsilon \ell_e,\epsilon \ell_e)$ 
for $e \in \E_{j,{\rm int}}$. Assume {\em a priori} that for all large $\epsilon$, we have 
\begin{equation}
\label{apriori-assumption}
\max_{e \in \E_{j,{\rm int}}} |a_{\epsilon}| \leq a_{j,{\rm max}},
\end{equation}  
where the constant $a_{j,{\rm max}} > 0$ is $\epsilon$-independent. 

By Lemma \ref{lemma-1}, the Neumann data at $v_j$ directed away from 
$\Gamma_{\epsilon} \backslash \E_{N,\epsilon}$ is 
\begin{equation}
q_j^{(1)} = D_j p_j + \mathcal{O}(\| \vec{p} \| e^{-\epsilon \ell_{\rm min}} 
+ \| \vec{p} \|^3),
\end{equation}
where $\mathcal{O}$ denotes the error terms in the bound (\ref{eq:nlin_DTN_value}). 

By Lemma \ref{lemma-2}, the Neumann data at $v_j$ directed away from $\E_{N,\epsilon}$ is 
\begin{eqnarray}
\nonumber
q_j^{(2)} & = & (K_j + 2L_j + M_j) p_j - 4 \sum_{e \in \E_{j,{\rm pend}}} e^{-\epsilon \ell_e}  
- 8 \sum_{e \in \E_{j,{\rm loop}}} e^{-\epsilon \ell_e}  \\
&& 
- 4 \sum_{e \in \E^-_{j,{\rm int}}} e^{-\epsilon \ell_e-a_e}
- 4 \sum_{e \in \E^+_{j,{\rm int}}} e^{-\epsilon \ell_e+a_e}
+ \mathcal{O}(\epsilon e^{-3 \epsilon \ell_{\rm j, min}}),
\end{eqnarray}
where 
$$
\ell_{j,{\rm min}} := \min\{\min_{e \in \E_{j,{\rm pend}}}\{ \ell_e\},\min_{e \in \E_{j,{\rm loop}}}\{ \ell_e\},\min_{e \in \E_{j,{\rm int}}}\{ \ell_e \}\}
$$ 
and $\mathcal{O}$ denotes the error terms in the bound (\ref{eq:single-bump-DtN}).
Note that the summation in $q_j^{(2)}$ includes two contributions from the looping edges in $\E_{j,{\rm loop}}$ due to the two ends of the looping edge incident to the vertex $v_j$.

The flux boundary condition in (\ref{kbc}) gives $q_j^{(1)} + q_j^{(2)} = 0$, which becomes the implicit equation on $p_j$. 
Since $C^1$ property of the error terms is proven in \cite{BMP} 
(Remark \ref{remark-1}), the implicit equation is immediately solved 
with 
\begin{eqnarray}
\nonumber
p_j & = & \frac{4}{Z_j}  \left( 
 \sum_{e \in \E_{j,{\rm pend}}} e^{-\epsilon \ell_e}  
+ 2 \sum_{e \in \E_{j,{\rm loop}}} e^{-\epsilon \ell_e} 
+ \sum_{e \in \E^-_{j,{\rm int}}} e^{-\epsilon \ell_e-a_e} 
+ \sum_{e \in \E^+_{j,{\rm int}}} e^{-\epsilon \ell_e+a_e}
\right) \\
&& + \mathcal{O}\left( \| \vec{p} \| e^{-\epsilon \ell_{\rm min}} + 
+ \| \vec{p} \|^3 + \epsilon e^{-3 \epsilon \ell_{\rm j, min}} \right),
\label{solution-pj}
\end{eqnarray}
where $Z_j := D_j + K_j + 2 L_j + M_j$ is the total degree of the vertex $v_j$ 
in $\Gamma_{\epsilon}$. 

\begin{remark}
	\label{remark-conditions}
Although the boundary conditions (\ref{kbc}) are satisfied 
for all vertices if $p_j$ is defined by (\ref{solution-pj}) 
for $1 \leq j \leq |B|$, there are two consistency conditions 
to be verified. One condition arises from the fact that 
$\| \vec{p}\|$ in the error terms in (\ref{solution-pj})
include all Dirichlet data $\vec{p} = (p_1,p_2,\dots,p_{|B|})$.
Therefore, the error terms in (\ref{solution-pj}) must be smaller 
than the leading-order terms in (\ref{solution-pj}) as $\epsilon \to \infty$ 
for each $j$. The second condition arises from the fact that the values $a_e$ for $e \in \E_{j,{\rm int}}$ must be uniquely defined and be bounded as $\epsilon \to \infty$ according to the {\em a priori} assumption (\ref{apriori-assumption}). 
\end{remark}

In order to complete the proof of Theorem \ref{theorem-main1}, it remains to verify the two consistency conditions in Remark \ref{remark-conditions} from the two technical assumptions, 
Assumptions \ref{assumption-non-degeneracy} and \ref{assumption-internal-edge}.

The first consistency condition is satisfied if 
$$
\| \vec{p} \| e^{-\epsilon \ell_{\rm min}} \ll e^{-\epsilon \ell_{j,{\rm min}}}, 
\quad \| \vec{p} \|^3 \ll e^{-\epsilon \ell_{j,{\rm min}}}, \quad 
1 \leq j \leq |B|, 
$$
where $e^{-\epsilon \ell_{j,{\rm min}}}$ defines 
the size of the leading-order terms in (\ref{solution-pj}). This leads to the constraints (\ref{constraint-on-length-1}) and (\ref{constraint-on-length-2}) in 
Assumption \ref{assumption-non-degeneracy}.

The second consistency condition can be formulated in terms of 
the invariant curve (\ref{energy-level}). 
The same level $\beta$ for the large solution of Lemma \ref{lemma-2} 
on a given internal edge must correspond to two segments extending to two different boundary vertices.

Let $e_0 \in \E_{N,\epsilon}$ be an internal edge connecting 
$v_j \in \V_B$ at the left end and $v_k \in \V_B$ at the right end, 
so that $e_0 \in \E^-_{j,{\rm int}}$ and $e_0 \in \E^+_{k,{\rm int}}$ 
and $j \neq k$. 
Then, $a_{e_0}$ must be found from the following equation 
\begin{equation}
\label{equation-on-a}
q_{e_0 \sim v_j}^2 - p_j^2 + p_j^4 = q_{e_0 \sim v_k}^2 - p_k^2 + p_k^4,
\end{equation}
where $q_{e_0 \sim v_j}$ and $q_{e_0 \sim v_k}$ are given by the expansion 
\begin{equation}
\label{solution-qj}
\left\{ \begin{array}{l} 
q_{e_0 \sim v_j} = 4 e^{-\epsilon \ell_{e_0}-a_{e_0}} - p_j + \mathcal{O}(\epsilon e^{-3 \epsilon \ell_{e_0}}), \\
q_{e_0 \sim v_k} = -4 e^{-\epsilon \ell_{e_0} + a_{e_0}} + p_k + \mathcal{O}(\epsilon e^{-3 \epsilon \ell_{e_0}}),
\end{array}
\right.
\end{equation}
due to the bounds (\ref{eq:single-bump-DtN}).

Let us now assume that the internal edge $e_0 \in \E_{N,\epsilon}$ 
does not have common boundary vertices with other internal edges in $\E_{N,\epsilon}$ as in Assumption \ref{assumption-internal-edge}. 
This means that 
$\E_{j,{\rm int}} = \{ e_0\}$ and $\E_{k,{\rm int}} = \{e_0\}$.
Substituting (\ref{solution-pj}) and (\ref{solution-qj}) into 
the two sides of equation (\ref{equation-on-a}) gives the expansions:
\begin{eqnarray*}
&& q_{e_0 \sim v_j}^2 - p_j^2 + p_j^4 
= \frac{16}{Z_j} \left[ (Z_j-2) e^{-2\epsilon \ell_{e_0}-2a_{e_0}} 
- 2 e^{-\epsilon \ell_{e_0} - a_{e_0}} \left( 
 \sum_{e \in \E_{j,{\rm pend}}} e^{-\epsilon \ell_e}  
+ 2 \sum_{e \in \E_{j,{\rm loop}}} e^{-\epsilon \ell_e}  \right) \right] \\
&& 
+ \mathcal{O}\left( \| \vec{p} \| e^{-\epsilon \ell_{\rm min} - \epsilon \ell_0} + 
+ \| \vec{p} \|^3 e^{-\epsilon \ell_0} + \epsilon e^{-3 \epsilon \ell_{j,{\rm min}} - \epsilon \ell_{e_0}} + e^{-4 \epsilon \ell_{j,{\rm min}}} 
+ \epsilon e^{-3 \epsilon \ell_{e_0}-\epsilon \ell_{j,{\rm min}}}
	+ \epsilon  e^{-4 \epsilon \ell_{e_0}} \right) 
\end{eqnarray*}
and
\begin{eqnarray*}
	&& q_{e_0 \sim v_k}^2 - p_k^2 + p_k^4 
= \frac{16}{Z_k} \left[ (Z_k-2) e^{-2\epsilon \ell_{e_0} + 2a_{e_0}} 
	- 2 e^{-\epsilon \ell_{e_0} + a_{e_0}} \left( 
	\sum_{e \in \E_{k,{\rm pend}}} e^{-\epsilon \ell_e}  
	+ 2 \sum_{e \in \E_{k,{\rm loop}}} e^{-\epsilon \ell_e}  \right) \right] \\
	&& 
	+ \mathcal{O}\left( \| \vec{p} \| e^{-\epsilon \ell_{\rm min} - \epsilon \ell_0} + 
	+ \| \vec{p} \|^3 e^{-\epsilon \ell_0} + \epsilon e^{-3 \epsilon \ell_{k,{\rm min}} - \epsilon \ell_{e_0}} 
	+ e^{-4 \epsilon \ell_{k,{\rm min}}} 
	+ \epsilon e^{-3 \epsilon \ell_{e_0}-\epsilon \ell_{k,{\rm min}}}
	+ \epsilon  e^{-4 \epsilon \ell_{e_0}} \right).
\end{eqnarray*}
By Remark \ref{remark-degree}, we have $Z_j \geq 3$ and $Z_k \geq 3$ for the internal edge $e_0$. Equating the two sides in (\ref{equation-on-a}) 
and dividing by $e^{-2\epsilon \ell{e_0}}$
gives the unique solution of the implicit equation:
\begin{eqnarray}
\label{solution-a}
e^{4a_{e_0}} = \frac{(Z_j-2) Z_k}{(Z_k-2) Z_j} + E_0,
\end{eqnarray}
where the error terms in $E_0$ are given by 
\begin{eqnarray*}
E_0 &=& \frac{2}{Z_k-2} e^{3 a_{e_0} + \epsilon \ell_{e_0}} 
\left( 
\sum_{e \in \E_{k,{\rm pend}}} e^{-\epsilon \ell_e}  
+ 2 \sum_{e \in \E_{k,{\rm loop}}} e^{-\epsilon \ell_e}  \right) \\
\nonumber
&& -\frac{2 Z_k}{Z_j(Z_k-2)} e^{a_{e_0} + \epsilon \ell_{e_0}} 
\left( 
\sum_{e \in \E_{j,{\rm pend}}} e^{-\epsilon \ell_e}  
+ 2 \sum_{e \in \E_{j,{\rm loop}}} e^{-\epsilon \ell_e}  \right) \\
&& 
+ \mathcal{O}\left( \| \vec{p} \| e^{\epsilon \ell_{e_0} -\epsilon \ell_{\rm min}} + 
+ \| \vec{p} \|^3 e^{\epsilon \ell_{e_0}} + 
\epsilon e^{\epsilon \ell_{e_0} -3 \epsilon \ell_{j,{\rm min}}}  + 
\epsilon e^{\epsilon \ell_{e_0} -3 \epsilon \ell_{k,{\rm min}}} \right) \\
&&	
+ \mathcal{O}\left(  e^{2\epsilon \ell_{e_0}  - 4 \epsilon \ell_{j,{\rm min}}} + e^{2\epsilon \ell_{e_0}  - 4 \epsilon \ell_{k,{\rm min}}} 
+ \epsilon e^{-\epsilon \ell_{e_0}-\epsilon \ell_{j,{\rm min}}}
+ \epsilon e^{-\epsilon \ell_{e_0}-\epsilon \ell_{k,{\rm min}}}
+ \epsilon  e^{-2 \epsilon \ell_{e_0}}  \right).
\end{eqnarray*}
The error terms beyond the first (leading-order) term in (\ref{solution-a}) are small if 
\begin{eqnarray*}
	&& \| \vec{p} \| e^{-\epsilon \ell_{\rm min}} \ll e^{-\epsilon \ell_{e_0}}, \\
&& \| \vec{p} \|^3 \ll e^{-\epsilon \ell_{e_0}}, \\
&& \epsilon e^{-3 \epsilon \ell_{j,{\rm min}}} \ll e^{-\epsilon \ell_{e_0}}, \\
&& \epsilon e^{-3 \epsilon \ell_{k,{\rm min}}} \ll e^{-\epsilon \ell_{e_0}},\\
&& e^{-\epsilon \ell_{e}} \ll e^{-\epsilon \ell_{e_0}}, \quad 
\forall e \in \E_{j,{\rm pend}} \cup \E_{j, {\rm loop}}, \\
&& e^{-\epsilon \ell_{e}} \ll e^{-\epsilon \ell_{e_0}}, \quad 
\forall e \in \E_{k,{\rm pend}} \cup \E_{k, {\rm loop}}.
\end{eqnarray*}
The first two conditions coincide with the consistency conditions satisfied by 
Assumption \ref{assumption-non-degeneracy}. 
The other four conditions are satisfied if 
$\ell_{e_0}$ is strictly minimal among the half-lengths of the looping edges and the lenghts of the pendant edges as in Assumption \ref{assumption-internal-edge}. 
This means that 
$$
\ell_{e_0} = \ell_{j,{\rm min}} = \ell_{k,{\rm min}} < \ell_{e}, \quad 
\forall e \in \E_{j,{\rm pend}} \cup \E_{j, {\rm loop}} \cup \E_{k,{\rm pend}} \cup \E_{k, {\rm loop}}. 
$$
Since the first (leading-order) term in (\ref{solution-a}) is 
independent of $\epsilon$, the {\em a priori} assumption (\ref{apriori-assumption}) on the admissible values of $a_{e_0}$ is satisfied by the solution (\ref{solution-a}).

The proof of Theorem \ref{theorem-main1} is complete. The bound (\ref{eq:L2}) follows from the bound (2.24) in Theorem 2.9 in \cite{BMP}. 

\begin{remark}
	If two or more internal edges in $\E_N$ are connected to the same boundary vertex, a system of two or more equations (\ref{equation-on-a}) is set up and the solution for $a_{e_0}$ is no longer available in the simple form (\ref{solution-a}). It is generally hard to obtain the solution for $a_{e_0}$ on each internal edge in the system of nonlinear equations. 
\end{remark}

\section{Morse index for multi-pulse edge-localized states}
\label{sec-proof-2}

Here we count the Morse index for multi-pulse edge-localized states and give the proof of Theorem \ref{theorem-main2}.

Let $\Phi$ be the positive $N$-pulse edge-localized state of Theorem \ref{theorem-main1} such that the set $\E_N$ only contains pendant and 
looping edges. Let $\mathcal{L} : H^2_{\rm NK}(\Gamma) \subset L^2(\Gamma) \to L^2(\Gamma)$ be the linearized operator with the differential expression given by (\ref{Lplus}). The number of its negative eigenvalues (the Morse index)  is denoted by $n(\mathcal{L})$ and the multiplicity of its zero eigenvalue is denoted by $z(\mathcal{L})$. We will prove that $n(\mathcal{L}) = N$ and $z(\mathcal{L}) = 0$ for large enough $\epsilon$.

Applying the scaling transformation 
$\Phi(x) = \epsilon U(z)$ with $z = \epsilon x$ and $\epsilon := \sqrt{-\omega}$, we rewrite the spectral problem for $\calL$ as a collection of differential equations
\begin{equation}
\label{sp-on-edges}
-w_j''(z) + w_j(z) - 6 u_j(z)^{2} w_j(z) = \lambda w_j(z),
\quad z \in e_{j, \epsilon}, 
\end{equation}
subject to the boundary conditions (\ref{kbc}) at vertices $v \in \V$. Here $W = (w_1, w_2, \dots, w_{|\E|})$ is a rescaled eigenvector in $H^2_{\rm NL}(\Gamma_{\epsilon})$ 
and $\lambda$ is a rescaled eigenvalue of the linearized operator $\mathcal{L}$. Abusing notations, we still refer 
to the linearized operator in the rescaled variables as to $\mathcal{L}$.

\subsection{Idea of the proof}

Recall that $\V_B$ is the set of boundary vertices separating 
$\E_{N,\epsilon}$ and $\Gamma_{\epsilon} \backslash \E_{N,\epsilon}$. 
Then for every $\bm{\alpha} = (\alpha_1, \alpha_2, \cdots, \alpha_{|B|})$, 
let the space $H^2_{\bm{\alpha}}(\Gamma_\epsilon)$ be defined by the modified boundary conditions:
\begin{equation}
\label{kbc-alpha}
\left\{ \begin{array}{l}
W \text{ is continuous on } \Gamma_\epsilon, \\
\sum_{e \sim v} \partial W(v) = 0
\text{ for every vertex } 
v \in \V \backslash \V_B, \\
\sum_{e \sim v_j} \partial W(v_j) = \alpha_j W(v_j)
\text{ for every vertex } 
v_j \in \V_B,
\end{array} \right.
\end{equation}
where the derivatives $\partial$ are directed away from the vertex.

We use $\bm{\alpha} = 0$ to denote $\bm{\alpha} = (0, 0, \dots, 0)$ and $\bm{\alpha} = \infty$ to denote $\bm{\alpha} = (+\infty, + \infty, \dots, +\infty)$. Then, we have $H^2_{\rm NK}(\Gamma_{\epsilon}) \equiv H^2_{\bm{\alpha} = 0}(\Gamma_{\epsilon})$ and $H^2_{\rm D}(\Gamma_{\epsilon}) \equiv H^2_{\bm{\alpha} = \infty}(\Gamma_{\epsilon})$, where $H^2_{\rm D}(\Gamma_{\epsilon})$ 
stands for the domain with the Dirichlet conditions at the boundary vertices:
\begin{equation}
\label{kbc-infty}
\left\{ \begin{array}{l}
W \text{ is continuous on } 
\Gamma_\epsilon, \\
\sum_{e \sim v} \partial W(v) = 0
\text{ for every vertex } 
v \in \V \backslash \V_B, \\
W(v_j) = 0
\text{ for every vertex } 
v_j \in \V_B.
\end{array} \right.
\end{equation}
We also introduce the linearized operator $\Lcal_{\bm{\alpha}}$, 
such that the spectral problem $\Lcal_{\bm{\alpha}} W = \lambda W$ is still represented 
by the differential equations (\ref{sp-on-edges}), but the domain of $\Lcal_{\bm{\alpha}}$ is $H^2_{\bm{\alpha}} (\Gamma_\epsilon) \subset L^2(\Gamma_{\epsilon})$.  

The proof of Theorem \ref{theorem-main2} relies 
on the continuity argument from $\bm{\alpha} = \infty$ to $\bm{\alpha} = 0$ 
given by the following proposition.

\begin{proposition}
\label{morse-index-equality}
Assume that $\mathcal{L}_{\bm{\alpha}} W  = 0$ admits no solutions $W \in H^2_{\bm{\alpha}}(\Gamma_{\epsilon})$ for every $\bm{\alpha} \in [0, \infty)^{|B|}$. 
Then, 
$$
n(\calL_{\bm{\alpha} = 0}) = n(\calL_{\bm{\alpha} = \infty}), \qquad z(\calL_{\bm{\alpha} = 0}) = z(\calL_{\bm{\alpha} = \infty}).
$$
\end{proposition}

\begin{proof}
The proof relies on the standard perturbation theory of linear operators given in Chapter 7 of \cite{Kato}. In short, since $\calL_{\bm{\alpha}}$ is a holomorphic family of self-adjoint operators, each isolated eigenvalue of $\calL_{\bm{\alpha}}$ depends continuously on $\alpha$. Therefore, if no eigenvalues of $\calL_{\bm{\alpha}}$ cross zero when $\bm{\alpha}$ 
traverses from $\bm{\alpha} = \infty$ to $\bm{\alpha}=0$, then the number of negative eigenvalues as well as the multiplicity of the zero eigenvalue remains the same. 
\end{proof}

By using Proposition \ref{morse-index-equality}, Theorem \ref{theorem-main2} holds if we can verify the following claims for large $\epsilon$:

$$
\begin{aligned}
& \text{\bf Claim 1:}\qquad  n(\calL_{\bm{\alpha} = \infty}) = N \text{   and   } z(\calL_{\bm{\alpha} = \infty}) = 0,\\
& \text{\bf Claim 2:}\qquad \mathcal{L}_{\bm{\alpha}} W  = 0 \text{ admits no solutions } W \in H^2_{\bm{\alpha}}(\Gamma_{\epsilon}) \text{ with } \alpha \in [0,\infty)^{|B|}.
\end{aligned}
$$
Below we give proofs of these claims. 

\subsection{Proof of Claim 1} 

The proof is based on the Sturm's Comparison theorem 
(see Section 5.5 in \cite{Teschl}). For simplicity, we assume that $\E_N$ contains only looping edges. Due to the symmetry of solutions in the looping edges, the proof extends to the pendant edges with minor modifications.

The spectral problem for the operator $\calL_{\bm{\alpha} = \infty}$ can be expressed by the second-order differential equations
(\ref{sp-on-edges}) equipped with the Dirichlet conditions at the boundary vertices given by (\ref{kbc-infty}).
The spectral problem for $\calL_{\bm{\alpha} = \infty}$ is given by $N$ uncoupled problems on the $N$ edges in $\E_{N,\epsilon}$ and another uncoupled problem in $\Gamma_{\epsilon} \backslash \E_{N,\epsilon}$. On each edge $e_{j, \epsilon} \in \E_{N,\epsilon}$, we have
\begin{equation}
\label{dirichlet-bvp-looping}
\left\{ \begin{array}{l}
-w_j''(z) + w_j(z) - 6 u_j(z)^{2} w_j(z) = \lambda w_j(z),
\quad z \in (-\epsilon \ell_j,\epsilon \ell_j), \\
w_j(- \epsilon \ell_j) = w_j( \epsilon \ell_j ) = 0,
\end{array} \right.
\end{equation}
where $u_j(z)$ satisfies the condition (\ref{bound-up}). 
Since $\epsilon$ is large, $p := u(\epsilon \ell_j)$ and $q := -u'(\epsilon \ell_j)$ are small and positive by Lemma \ref{lemma-2}. 
For $\lambda = 0$, the homogeneous equation (\ref{dirichlet-bvp-looping}) is solved by the superposition formula (\ref{superposition}) in Lemma  \ref{bound-up-s-vanishing-points}, where 
the even solution $s(z)$ has exactly one nodal point on $(0,\epsilon \ell_j)$ 
and the odd solution $u'(z)$ has no nodal points on $(0,\epsilon \ell_j)$.  By Sturm's Comparison theorem, the spectral problem (\ref{dirichlet-bvp-looping}) admits exactly one simple negative eigenvalue and no zero eigenvalue. Since eigenvalues of all $N$ spectral problems (\ref{dirichlet-bvp-looping}) also appear in the spectrum of the operator $\calL_{\bm{\alpha} = \infty}$, we have $n(\calL_{\bm{\alpha} = \infty}) \geq N$.

\begin{remark}
	If $e_{j,\epsilon} \in \E_{N,\epsilon}$ is a pendant edge, the Neumann condition at one end implies that the spectral problem (\ref{dirichlet-bvp-looping}) is solved in the space of even functions, where it still admits exactly one simple negative eigenvalue and no zero eigenvalue.
\end{remark}

Next, we show that the uncoupled problem in $\Gamma_\epsilon \backslash 
\E_{N, \epsilon}$ given by
\begin{equation}
\label{spectral-Dirichlet}
-W'' + W - 6 U^2 W = \lambda W, \quad 
z\in \Gamma_\epsilon \backslash 
\E_{N, \epsilon},
\end{equation}
subject to the boundary conditions (\ref{kbc-infty})  has a strictly positive spectrum. 
Since $U$ on  $\Gamma_\epsilon \backslash 
\E_{N, \epsilon}$ satisfies the bound (\ref{eq:nlin_DTN_solution}) 
with small $\vec{p}$ given by (\ref{solution-pj}), 
we obtain from (\ref{spectral-Dirichlet}) after multiplication by $W$ and integration over $\Gamma_\epsilon \backslash \E_{N, \epsilon}$ 
with the boundary conditions (\ref{kbc-infty}):
$$
\lambda \|W\|_{L^2(\Gamma_\epsilon \backslash 
	\E_{N, \epsilon})} = \|W\|_{L^2(\Gamma_\epsilon \backslash 
\E_{N, \epsilon})} - 6 \| U W \|_{L^2(\Gamma_{\epsilon} \backslash \E_{N,\epsilon})} + \|W'\|_{L^2(\Gamma_\epsilon \backslash 
\E_{N, \epsilon})} \geq C \| W \|^2_{L^2(\Gamma_{\epsilon} \backslash \E_{N,\epsilon})},
$$ 
for some $C > 0$. Hence, $\lambda \geq C > 0$ for every solution $W \in H^2_{\rm D}((\Gamma_{\epsilon} \backslash \E_{N,\epsilon})$, and the spectral problem (\ref{spectral-Dirichlet}) does not contribute into the nonpositive part of the spectrum of the operator $\calL_{\bm{\alpha} = \infty}$. This completes the proof of $n(\calL_{\bm{\alpha} = \infty}) = N$ and  $z(\calL_{\bm{\alpha} = \infty}) = 0$.

\subsection{Proof of Claim 2}

Let $W$ be a solution of $\calL_{\bm{\alpha}} W = 0$ in $H^2(\Gamma_\epsilon)$ with large enough $\epsilon$. 
For every looping edge $e_{j,\epsilon} \in \E_{N,\epsilon}$ associated with the segment $[-\epsilon \ell_j, \epsilon \ell_j]$, the restriction of $W$ to the edge $e_{j,\epsilon}$ denoted by $w_j$ solves the second-order differential equation 
\begin{equation}
\label{linearized-equation-on-the-edge}
-w_j''(z) + w_j(z) - 6 u_j(z)^2 w_j(z) = 0, \quad 
z\in (-\epsilon \ell_j, \epsilon \ell_j).
\end{equation}
where $u_j(z)$ satisfies the condition (\ref{bound-up}).
The following lemma computes the 
contributions of the solution of the differential equation 
(\ref{linearized-equation-on-the-edge}) to the last boundary condition 
in (\ref{kbc-alpha}).

\begin{lemma}
\label{lemma-looping-edge-alpha}
If $e_{j,\epsilon} \in \E_{N, \epsilon}$ is a looping edge with $\epsilon$ large enough, then either $w_j \equiv 0$ on the entire edge or
\begin{equation}
\label{alpha-pm}
\alpha_\pm (e_{j,\epsilon}) := \mp \frac{w_j'(\pm \epsilon \ell_j)}{w_j(\pm \epsilon \ell_j)} < 0.
\end{equation}
\end{lemma}

\begin{proof}
	By Lemma \ref{bound-up-s-vanishing-points}, the general solution is 
	given by (\ref{superposition}) rewritten again as 
	\begin{equation}
	\label{w-decomposition}
	w_j(z) = A u_j'(z) + B s_j(z)
	\end{equation}
	for some parameters $A$ and $B$. 
	Since $u_j'(z)$ is odd, $s_j(z)$ is even, and $u_j'(\pm \epsilon \ell_j) \neq 0$, the continuity condition $w_j(-\epsilon \ell_j) = w_j(\epsilon \ell_j)$ necessarily implies $A = 0$. 
	
	If $B = 0$, then $w_j \equiv 0$ on the entire edge. If $B \neq 0$, 
	then $w_j(\pm \epsilon \ell_j) \neq 0$ and $\alpha_{\pm}(e_{j,\epsilon})$ in (\ref{alpha-pm}) 
	is defined by $s_j(\pm \epsilon \ell_j)$ and $s_j'(\pm \epsilon \ell_j)$. It follows 
	from the second inequality in (\ref{bound-up-weyl-tit-s}) and the symmetry of $s_j(z)$ that $\alpha_{\pm}(e_{j,\epsilon}) < 0$.
\end{proof}

\begin{remark}
	If $e_{j,\epsilon} \in \E_{N,\epsilon}$ is a pendant edge, the Neumann condition at one end still implies that $A = 0$, after which the proof of Lemma \ref{lemma-looping-edge-alpha} holds and gives $\alpha_+(e_{j,\epsilon}) < 0$.
\end{remark}

We denote the value of $W$ at each boundary point $v_j \in \V_B$ as 
$\mathfrak{p}_j := W(v_j)$. Then on the subgraph $\Gamma_\epsilon \backslash \E_{N, \epsilon}$, $W$ solves the homogeneous equation
\begin{equation}
\label{spectral-alpha}
-W'' + W - 6 U^2 W = 0, \quad 
z\in \Gamma_\epsilon \backslash \E_{N, \epsilon},
\end{equation}
subject to the nonhomogeneous Dirichlet boundary conditions:
\begin{equation}
\label{kbc-infty-complement}
\left\{ \begin{array}{l}
W \text{ is continuous on } 
\Gamma_\epsilon  \backslash \E_{N, \epsilon}, \\
\sum_{e \sim v} \partial W(v) = 0
\text{ for every vertex } 
v \in \V \backslash \V_B, \\
W(v_j) = \mathfrak{p}_j
\text{ for every vertex } 
v_j \in \V_B.
\end{array} \right.
\end{equation}
The following lemma computes the 
contributions of the solution of the boundary-value problem 
(\ref{spectral-alpha})--(\ref{kbc-infty-complement}) 
to the last boundary condition 
in (\ref{kbc-alpha}).

\begin{lemma}
\label{lemma-complement-alpha}
There are $C_0 > 0$ and $\epsilon_0 > 0$ such that for every $\vec{\mathfrak{p}} = (\mathfrak{p}_1,\mathfrak{p}_2,\dots,\mathfrak{p}_{|B|})$ and every $\epsilon > \epsilon_0$, the unique solution $W \in H^2(\Gamma_{\epsilon} \backslash \E_{N,\epsilon})$ of the boundary-value problem (\ref{spectral-alpha})--(\ref{kbc-infty-complement}) 
satisfies 
\begin{equation}
\label{alpha-W}
\| W \|_{H^2(\Gamma_{\epsilon} \backslash \E_{N,\epsilon})} \leq C_0 \| \vec{\mathfrak{p}}\|
\end{equation}
and
\begin{equation}
\label{alpha-complement}
|\mathfrak{q}_j - D_j \mathfrak{p}_j| \leq C_0 \| \vec{\mathfrak{p}} \| e^{-\epsilon \ell_{\rm min}},
\end{equation}
where $\mathfrak{q}_j$ is the Neumann data directed away from  $\Gamma_{\epsilon} \backslash \E_{N,\epsilon}$, $D_j$ is the degree of the $j$-th boundary vertex $v_j \in \V_B$ in $\Gamma_{\epsilon} \backslash \E_{N,\epsilon}$, and $\ell_{\rm min}$ is the length of the shortest edge in $\Gamma \backslash \E_N$.
\end{lemma}

\begin{proof}
Since $\| U \|_{L^{\infty}(\Gamma_{\epsilon} \backslash \E_{N,\epsilon})}$ is small for large enough $\epsilon$ by the bound (\ref{eq:nlin_DTN_solution}) 
with $\vec{p}$ given by (\ref{solution-pj}), the boundary-value problem (\ref{spectral-alpha})--(\ref{kbc-infty-complement}) can be analyzed with Theorem 2.1 and Lemma 2.3 in \cite{BMP} for $g = 0$. The bound (\ref{alpha-complement}) follows from the bound (2.4) in \cite{BMP} due to the smallness of $\| U \|_{L^{\infty}(\Gamma_{\epsilon} \backslash \E_{N,\epsilon})}$.
\end{proof}

\begin{remark}
	\label{remark-Neumann}
The Neumann data $\mathfrak{q}_j$ in (\ref{alpha-complement}) is directed in the opposite direction compared to $\partial \Psi(v_j)$ in (\ref{kbc-alpha}). In other words, 
\begin{equation}
\label{relation-neumann-data}
\sum_{ e \sim v_j, e \in \Gamma_\epsilon \backslash \E_{N,\epsilon}} \partial \Psi(v_j) = -\mathfrak{q}_j = -D_j \mathfrak{p}_j + \mathcal{O} (\| \vec{\mathfrak{p}} \| e^{\epsilon \ell_{\rm min}}).
\end{equation}
\end{remark}

We are now ready to show that no solutions $W \in H^2(\Gamma_{\epsilon})$ of 
$\mathcal{L}_{\bm{\alpha}} W = 0$ satisfy the last boundary condition in (\ref{kbc-alpha}) with $\alpha_j$ being nonnegative. 

First, we consider the case when $\mathfrak{p}_j = 0$ for all $v_j \in \V_B$. Then, the right hand side of the last boundary conditions in (\ref{kbc-alpha}) vanishes and $W$ satisfies the Neumann-Kirchhoff (NK) boundary conditions. Moreover, by Lemma \ref{lemma-looping-edge-alpha}, the 
$\E_{N,\epsilon}$-components of the solution $W$ are entirely zero, and make no contribution into the last condition in (\ref{kbc-alpha}). 
In what follows, the restriction of $W$ to $\Gamma_\epsilon \backslash \E_{N,\epsilon}$ solves the second-order differential equation (\ref{spectral-alpha}) subject to the Dirichlet conditions at the boundary 
vertices. By Lemma \ref{lemma-complement-alpha} due to the bound (\ref{alpha-W}), the only solution on $\Gamma_\epsilon \backslash \E_{N,\epsilon}$ is zero, 
hence $W = 0$ on $\Gamma_{\epsilon}$. 

Therefore, there exists at least one boundary vertex $v_j \in \V_B$ such that the corresponding $\mathfrak{p}_j \neq 0$. Up to multiplication by a constant, we may assume that $\mathfrak{p}_j = 1$. Using Lemmas \ref{lemma-looping-edge-alpha} and \ref{lemma-complement-alpha} and the expansion (\ref{relation-neumann-data}), the last condition in (\ref{kbc-alpha}) for large enough $\epsilon$ becomes
\begin{equation}
\label{count-alpha}
\alpha_j = \sum_{e \in \E_{j,\epsilon}} \alpha_-(e) + \sum_{e \in \E_{j,\epsilon}} \alpha_+(e) 
- \mathfrak{q}_j < 0,
\end{equation} 
where $\E_{j,\epsilon} \subset \E_{N,\epsilon}$ is the set of edges in $\E_{N,\epsilon}$ adjacent to the vertex $v_j \in \V_B$.
This contradicts to the assumption that $\alpha_j \in [0,\infty)$. Hence no 
$W \in H^2_{\bm{\alpha}}(\Gamma_{\epsilon})$ exists such that 
$\mathcal{L}_{\bm{\alpha}} W  = 0$ for $\alpha \in [0,\infty)^{|B|}$.

\section{Examples}
\label{sec-4}

We end this paper with three examples of quantum graphs, where the assumptions of Theorems \ref{theorem-main1} and \ref{theorem-main2} can be checked. We also discuss limitations of results of Theorems \ref{theorem-main1} and \ref{theorem-main2} to cover all edge-localized states on these graphs.

\begin{figure}[htbp]
	\centering
	\includegraphics[width=4in, height = 2in]{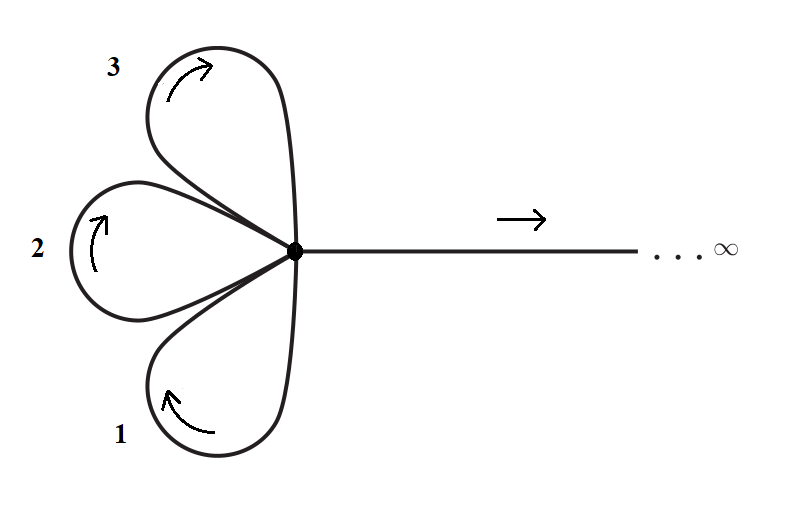}
	\caption{A flower graph with three loops as in Example \ref{ex-flower}. Each loop $j$ with $j=1,2,3$ is parametrized by a segment $[-\ell_j, \ell_j]$ of length $2 \ell_j$, and the unbounded edge is parametrized by a half-line $[0, \infty)$.
		The direction of the axis associated with each of the segments is specified by the arrows. }
	\label{fig-flower}
\end{figure}

\begin{example}
	\label{ex-flower}
{\rm	Consider a flower graph consisting of $L$ loops of the lengths $\{ 2 \ell_j \}_{j=1}^L$ and one half-line connected at a single common vertex, see Figure \ref{fig-flower} for $L = 3$. Assumption \ref{assumption-non-degeneracy} is satisfied 
	because $\ell_{\rm min} = \infty$ in (\ref{constraint-on-length-1})
	and $|B| = 1$ in (\ref{constraint-on-length-2}). Assumption \ref{assumption-internal-edge} is satisfied because no internal edges exist. 
	Theorem \ref{theorem-main1} states existence of the $N$-pulse edge-localized states for every $1 \leq N \leq L$ and Theorem \ref{theorem-main2} states that Morse index of this $N$-pulse state is $N$. This coincides with 
	the main results of \cite{KMPX} (Theorems 1, 2, and 3) in the limit $\omega \to -\infty$ obtained with long analysis of the period function in the case of loops of the same normalized length.}
\end{example}

\begin{example}
	\label{ex-dumbbell}
{\rm 	Consider a dumbbell graph consisting of two loops of lengths $2\ell_-$ and $2 \ell_+$ connected by an internal edge of length $2 \ell_0$ at two vertices, see Figure \ref{fig-dumbbell}. Assume that $\ell_- \leq \ell_+$ for convenience. We can construct several edge-localized states as follows.\\
	\begin{itemize}
		\item By Theorem \ref{theorem-main1}, assumptions of which are always satisfied if $N = 1$, the $1$-pulse state can be placed at any of the three edges in the limit $\omega \to -\infty$ in agreement with the main results of \cite{AST} and \cite{BMP}. By Theorem \ref{theorem-main2}, Morse index of the $1$-pulse state placed at the loop is $1$. Morse index for the $1$-pulse state placed at the internal edge is not defined by Theorem \ref{theorem-main2}. By the variational theory in \cite{AST}, the Morse index of this $1$-pulse state is also $1$. It was also proven explicitly in Lemma 4.10 in \cite{MP} for the symmetric dumbbell graph with $\ell_- = \ell_+$ that the Morse index of this $1$-pulse state is $1$.\\
		
		\item By Theorem \ref{theorem-main1}, the $2$-pulse state can be placed at the two loops if $\ell_+ < 2 \ell_0 + \ell_-$ and $\ell_+ < 3 \ell_-$ from the conditions (\ref{constraint-on-length-1}) and (\ref{constraint-on-length-2}) respectively. The constraints are satisfied if $\ell_- = \ell_+$. By Theorem \ref{theorem-main2}, the Morse index of this $2$-pulse state is $2$.\\
		
		\item By Theorem \ref{theorem-main1}, the $2$-pulse state can also be placed at one loop and the internal edge if $\ell_0 < \ell_-$ due to Assumption \ref{assumption-internal-edge}. The conditions (\ref{constraint-on-length-1}) and (\ref{constraint-on-length-2}) of Assumption \ref{assumption-non-degeneracy} are satisfied if $\ell_0 < \ell_- \leq \ell_+$. The Morse index of this $2$-pulse state 
		is not defined by Theorem \ref{theorem-main2}.\\
		
		\item By Theorem \ref{theorem-main1}, the $3$-pulse state can be placed at all three edges if $\ell_0 < \ell_- \leq \ell_+$. The Morse index of this $3$-pulse state is not defined by Theorem \ref{theorem-main2}.\\
	\end{itemize}
}
\end{example}

 \begin{figure}[htbp] 
	\centering
	\includegraphics[width=4in, height = 2in]{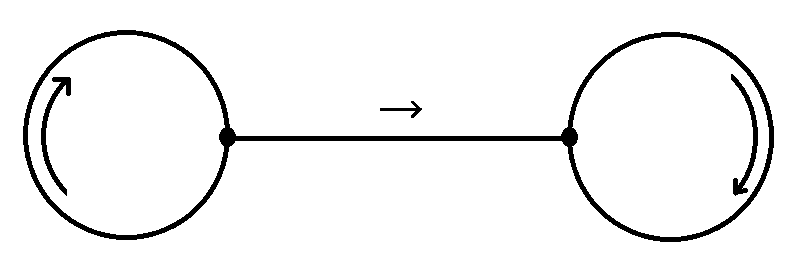}
	\caption{A dumbbell graph as in Example \ref{ex-dumbbell}. The left loop is parametrized by a segment $[-\ell_-, \ell_-]$ of length $2\ell_-$, and the right loop is parametrized by a segment $[-\ell_+, \ell_+]$ of length $2\ell_+$. The internal edge is parametrized by a segment $[-\ell_0, \ell_0]$ of length $2 \ell_0$. The direction of the axis associated with each of the segments is specified by the arrows. }
	\label{fig-dumbbell}
\end{figure}
	
\begin{example}
\label{example-counter}
{\rm Consider a graph $\Gamma$ which is given in Figure \ref{fig-line}. 
	It formally consists of three edges, where $e_1 = [0, \ell_1]$ and $e_3 = [0, \ell_3]$ are pendant edges and $e_2 = [-\ell_2, \ell_2]$ is the internal edge. However, $v_2$ and $v_3$ have are fake vertices of degree $2$, hence the same graph is mapped into a single interval $[-\ell_1-\ell_2,\ell_2+\ell_3]$ subject to Neumann boundary conditions at the end points. Assume that $\ell_1 \leq \ell_3$ for convenience.
	
	 	\begin{figure}[htbp] 
		\centering
		\includegraphics[width=2.8in, height = 0.7in]{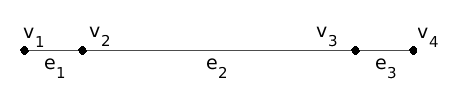}
		\caption{A segment graph $\Gamma$ in Example \ref{example-counter}.}
		\label{fig-line}
	\end{figure}

	All possible solutions of the Neumann boundary-value problem for the stationary equation $-u'' + u - 2 u^3 = 0$ are given by the integral curves shown on Fig. \ref{fig-phase-plane}. Due to the Neumann boundary conditions, the positive solution starts and ends at one of the two points $(p_-,0)$ and $(p_+,0)$ shown in red. Besides two simplest states loalized at one of the two pendant edges from a half-loop on Fig. \ref{fig-phase-plane}, we can construct three additional edge-localized states as follows.} \\

	\begin{figure}[htbp] 
	\centering
	\includegraphics[width=3.2in, height = 2.5in]{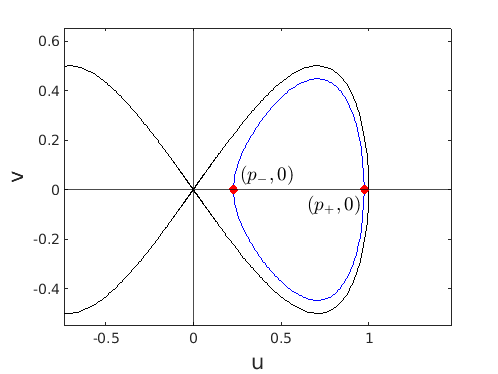}
	\caption{Phase plane $(u,v)$ of the differential equation $-u'' + u - 2 u^3 = 0$ showing the homoclinic loop for $\beta = 0$ and the integral curve $E_{\beta}$ for $\beta \in (-\frac{1}{4},0)$.}
	\label{fig-phase-plane}
\end{figure}

\begin{itemize}	
	\item 
	{\rm Let us take a single loop on the integral curve of Fig. \ref{fig-phase-plane} with the boundary conditions 
		$$
		u(-\epsilon (\ell_1+\ell_2)) = u(\epsilon(\ell_2+\ell_3)) = p_-. 
		$$
		Due to the spatial symmetry, the $1$-pulse state concentrates 
		at the middle point $z = \frac{1}{2} \epsilon (\ell_3 - \ell_1)$. 
		Since $\ell_1 \leq \ell_3$ implies $\ell_1 < 2 \ell_2 + \ell_3$, the middle point belongs to $e_2$, 
		hence the edge-localized state in the limit of large $\epsilon$ 
		(when $p_- \to 0$) corresponds to the choice $\mathcal{E}_1 : = \{e_2\}$ for the graph $\Gamma$ on Fig. \ref{fig-line}. This $1$-pulse state is illustrated on the left panel of Fig. \ref{fig-2-soliton}. The right panel shows the derivative $u'(z)$ which satisfies the homogeneous equation $-w''(z) + w(z) - 6 u(z)^2 w(z) = 0$ and Dirichlet conditions at the end points of the interval 
		$[-\epsilon (\ell_1+\ell_2),\epsilon(\ell_2+\ell_3)]$. 
		Using the Courant's nodal theorem (see Section 5.5 in \cite{Teschl}), since $u'(z)$ has two nodal domains on the interval, the zero eigenvalue is the second eigenvalue of the Dirichlet boundary-value problem (denoted as $\lambda_j^D$) interlacing with the eigenvalues of the Neumann boundary-value problem (denoted as $\lambda_j^N$) as follows:
\begin{equation}
\label{ordering-eigenvalues}
		\lambda_1^N < \lambda_1^D \leq \lambda^N_2 < \lambda^D_2 = 0 \leq \lambda_3^N < \lambda_3^D \leq \lambda_4^N < \dots
\end{equation}
Here the inequality between $\lambda^N_2$ and $\lambda^D_2$ is strict because 
the even potential $-6u(z)^2$ of the Schr\"{o}dinger operator 
$\partial_z^2 + 1 - 6 u(z)^2$ is extended periodically on $\mathbb{R}$ 
and it is a well-known fact that the anti-periodic eigenfunctions for $\lambda_1^D$ and $\lambda_2^N$ 
can not exist at the same eigenvalues as the periodic eigenfunctions 
for $\lambda^D_2$ and $\lambda_3^N$. Similarly, 
the inequality between $\lambda_3^N$ and $\lambda_3^D$ is strict. 
It follows from (\ref{ordering-eigenvalues}) that the Morse index of the $1$-pulse state is $n(\mathcal{L}) = 2$. This is a counter-example to conclusion of Theorem \ref{theorem-main2} with $N = 1$ which is the reason why the internal edges are excluded from the set $\E_N$ in Theorem \ref{theorem-main2}.}\\
	
	\begin{figure}[htbp] 
	\centering
	\includegraphics[width=3.2in, height = 2.5in]{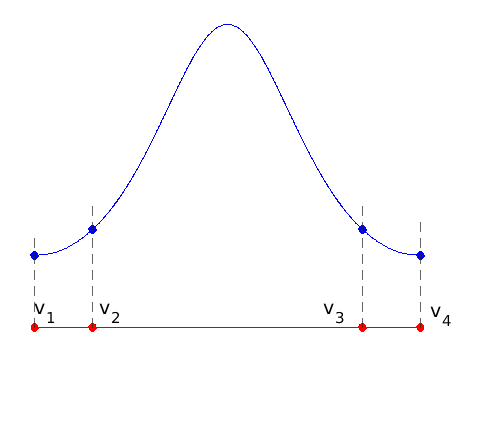}
	\includegraphics[width=3.2in, height = 2.5in]{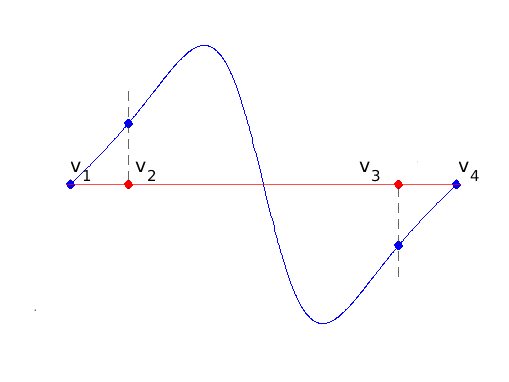}
	\caption{Left: the $1$-pulse state on the internal edge $e_2$. 
		Right: its derivative satisfying Dirichlet conditions on the interval.	}
	\label{fig-2-soliton}
\end{figure}

\item {\rm Let us take a single loop on the integral curve of Fig. \ref{fig-phase-plane} with the boundary conditions 
		$$
		u(-\epsilon (\ell_1+\ell_2)) = u(\epsilon(\ell_2+\ell_3)) = p_+. 
		$$
		In the limit of large $\epsilon$ (when $p_+ \to 1$), 
		this $2$-pulse state corresponds to the choice 
		$\mathcal{E}_2 : = \{e_1,e_3\}$ for the graph $\Gamma$ on Fig. \ref{fig-line}. Assumption \ref{assumption-non-degeneracy} is satisfied if $\ell_+ < 2 \ell_0 + \ell_-$ and $\ell_+ < 3 \ell_-$. The $2$-pulse state exists for sufficiently large $\epsilon$ by Theorem \ref{theorem-main1}. This solution is illustrated on the left panel of Fig. \ref{fig-2-pendants}. The right panel shows the derivative $u'(z)$ with two nodal domains, so that the ordering (\ref{ordering-eigenvalues}) suggests that $n(\mathcal{L}) = 2$. This agrees with the count of Theorem \ref{theorem-main2} with $N = 2$ since the internal edge $e_2$ is not in the set $\E_N$ and the fake vertices $v_2$ and $v_3$ are not obstruction to the construction of edge-localized states on the pendant edges.}\\

\item {\rm Let us take two loops on the integral curve of Fig. \ref{fig-phase-plane} with the boundary conditions 
	$$
	u(-\epsilon (\ell_1+\ell_2)) = u(\epsilon(\ell_2+\ell_3)) = p_-, 
	$$
	as shown on Fig. \ref{fig-2-pulse}.
Although $p_- \to 0$ as $\epsilon \to 0$, this $2$-pulse state does not correspond to the choice $\mathcal{E}_2 : = \{e_1,e_3\}$ in Theorem \ref{theorem-main1} because the edge-localized states at the pendants must be centered at the ends under the Neumann conditions. Consequently, 
the derivative $u'(z)$ also shown on Fig. \ref{fig-2-pulse} has four nodal domains suggesting by the Courant's nodal theorem that $n(\mathcal{L}) = 4$ which would contradict Theorem \ref{theorem-main2} with $N = 2$. However, as we explained, the main results 
are not applicable to the $2$-soliton state which is not localized at the ends of the pendants.}
	\begin{figure}[htbp] 
		\centering
		\includegraphics[width=3.2in, height = 2.5in]{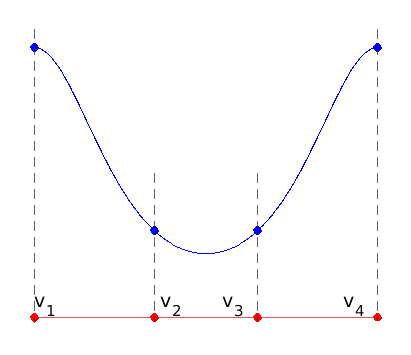}
		\includegraphics[width=3.2in, height = 2.5in]{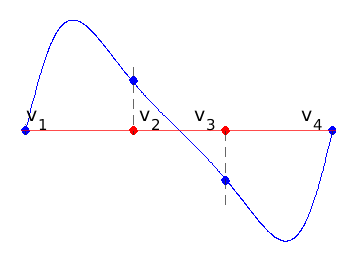}
		\caption{Left: the $2$-pulse state on the two pendants $e_1$ and $e_3$. 
			Right: its derivative satisfying Dirichlet conditions on the interval.	}
		\label{fig-2-pendants}
	\end{figure}	
	
	\begin{figure}[htbp] 
		\centering
		\includegraphics[width=3.2in, height = 2.5in]{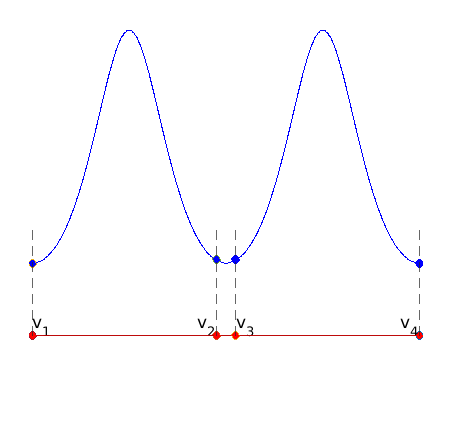}
		\includegraphics[width=3.2in, height = 2.5in]{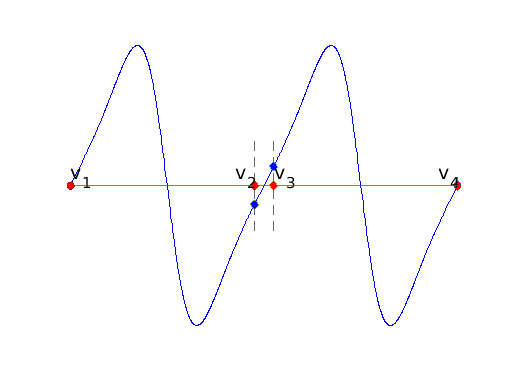}
		\caption{Left: the $2$-pulse state centered in the middle of the two pendants $e_1$ and $e_2$. Right: its derivative satisfying Dirichlet conditions on the interval.	}
		\label{fig-2-pulse}
	\end{figure}	

\end{itemize}
\end{example}

{\bf Acknowledgements:} A. Kairzhan thanks the Fields Institute for Research in Mathematical Sciences for its support and hospitality during the thematic program on Mathematical Hydrodynamics in July--December, 2020. D.E. Pelinovsky acknowledges the support from the NSERC Discovery grant.

\end{document}